\documentclass[10pt,conference]{IEEEtran}
\usepackage{mathpazo}
\usepackage{times}

\usepackage{amsmath}
\usepackage{amsfonts}
\usepackage{latexsym}
\usepackage{amssymb}

\usepackage{upref}
\usepackage{theorem}
\usepackage{graphicx}
\usepackage{psfrag}
\usepackage{algorithmic}
\usepackage{algorithm}
\usepackage{verbatim}
\usepackage{epstopdf}
\usepackage{cite}
\usepackage{color}

\theoremstyle{plain}
\theorembodyfont{\normalfont\slshape}

\newtheorem{thm}{Theorem$\!$}[section]
\newenvironment{theorem}
{\begin{thm}\hspace*{-1ex}{\bf.}}{\end{thm}}

\newtheorem{lem}{Lemma$\!$}[section]
\newenvironment{lemma}{\begin{lem}\hspace*{-1ex}{\bf.}}{\end{lem}}

\newtheorem{claimm}{Claim$\!$}[section]

\newtheorem{prop}{Proposition$\!$}[section]

\newtheorem{cor}{Corollary$\!$}[section]
\newenvironment{corollary}{\begin{cor}\hspace*{-1ex}{\bf.}}{\end{cor}}

\newtheorem{defn}{Definition$\!$}[section]
\newenvironment{definition}{\begin{defn}\hspace*{-1ex}{\bf.}}{\end{defn}}

\newtheorem{xmpl}{Example$\!$}[section]

\newtheorem{cnstr}{Construction$\!$}[section]

\newtheorem{fac}{Fact$\!$}[section]

\newenvironment{itemizei}{%
\begin{list}
  {\labelitemi}{ \leftmargin3em\labelwidth.5em\def\makelabel##1{\it{##1}}}}
{\end{list}}

\newcounter{enumrom}
\renewcommand{\theenumrom}{(roman{enumrom})}

\makeatletter
\renewcommand{\@endtheorem}{\endtrivlist}
\makeatother

\makeatletter
\renewcommand{\thefigure}{{\@arabic\c@figure}}
\renewcommand{\fnum@figure}{{\bf Figure\,\thefigure}}
\makeatother

\def\multiset#1#2{\ensuremath{\left(\kern-.3em\left(\genfrac{}{}{0pt}{}{#1}{#2}\right)\kern-.3em\right)}}
\usepackage{epsfig}
\usepackage{amsmath}
\usepackage{amsfonts}
\usepackage{amssymb}
\usepackage{graphics}
\usepackage{color}
\usepackage{epstopdf}
\usepackage{subfigure}

\date{}

\begin{document}

\title{A characterization of the number of subsequences obtained via the deletion channel}

\author{
\authorblockN{Y. Liron}
\authorblockA{The Open University of Israel\\
{\tt yuvalal@gmail.com}\vspace*{-4.0ex}}
\and
\authorblockN{M. Langberg}
\authorblockA{The Open University of Israel \\
{\tt mikel@openu.ac.il}\vspace*{-4.0ex}}
}

\date{}
\maketitle
\begin{abstract}
Motivated by the study of {\em deletion channels}, this work presents improved bounds on the number of subsequences obtained from a binary sting $X$ of length $n$ under $t$ deletions.
It is known that the number of subsequences in this setting strongly depends on the number of {\em runs} in the string $X$; where a run is a maximal sequence of the same character.
Our improved bounds are obtained by a structural analysis of the family of $r$-run strings $X$, an analysis in which we identify the {\em extremal} strings with respect to the number of subsequences.
Specifically, for every $r$, we present $r$-run strings with the minimum (respectively maximum) number of subsequences under any $t$ deletions; and perform an exact analysis of the number of subsequences of these extremal strings.
\end{abstract}

\section{Introduction}

Let $X \in \{0,1\}^n$ be a binary string of length $n$, and let $t \leq n$ be a parameter.
In this work, we study the size of the set $D_t(X)$ of subsequences of $X$ that can be obtained from $X$ via $t$ deletions.
The set $D_t(X)$ and its size play a major role in the design and analysis of communication schemes over {\em deletion channels}, i.e., channels in which characters of the transmitted codeword may be deleted,  \cite{Lev66,Lev01,MercierPHD,Mitz09,MBT10,Mitz11}.

The analysis of $D_t(X)$ is challenging as the number of subsequences of a string $X$ obtained by deletions does not depend only on its length $n$ and the number $t$ of deletions, but also strongly depends on its structure.
For example, $D_t(0^n)$ is of size 1 and equals the single string $0^{n-t}$, while there exist strings $X$ for which $D_t(X)$ is of size $\exp(\Omega(n-t))$.
Clearly, $|D_t(X)|$ is at most $2^{n-t}$ (as after $t$ deletions we remain with a binary string of length $n-t$).

In his work from 1966, Levenshtein \cite{Lev66} shows (as described in  \cite{Lev01}) that the number of subsequences $|D_t(X)|$ strongly depends on the number of {\em runs} in the string.
Here, a run is a maximal sequence of the same character, and the number of runs in a given string is denoted $r(\cdot)$. 
For example $r(0^n)=1$ while $r(0101\dots 01)=n$. 
Specifically, Levenshtein \cite{Lev66} proves that
$$
\binom{r(X)-t+1}{t}\le |D_t(X)| \le \binom{r(X)+t-1}{t}.
$$
Bounding $|D_t(X)|$ is addressed by Calabi and Hartnett \cite{Cal69}, which show that the maximal number of subsequences is obtained from certain strings $X$, denoted cyclic strings $C_n$, in which $r(X)=|X|$.
\cite{Cal69} devise a recursive expression for $|D_t(C_n)|$, to obtain the bound 
$$
\binom{r(X)-t+1}{t} \le |D_t(X)| \le  |D_t(C_n)|.
$$  
Relatively recently, Hirschberg and Regnier \cite{Hirsch} revisit the analysis of \cite{Cal69} and obtain an explicit upper bound together with an improved lower bound of the form  
$$
\sum_{i=0}^t \binom{r(X)-t}i \le |D_t(X)| \le \sum_{i=0}^t \binom{n-t}i.
$$  
Mercier et al. \cite{Bhar} study the setting of small values for $t$, and present explicit formulas for $D_t(X)$ for $t\le 5$.
However for general values of $t$ the problem remains open.
Several of the results above generalize also to arbitrary alphabets.

The bounds of \cite{Lev66,Cal69,Hirsch} are depicted in Figure~\ref{fig:previous} for the case $n=120$ and $r=r(X)=24$ as a function of $t$.
The lower bounds of both \cite{Hirsch} and \cite{Lev66} depend on the number of runs $r(X)$; and it holds that the lower bound of \cite{Hirsch} is superior (i.e., larger) to that of \cite{Lev66}.  
The upper bound of \cite{Lev66} depends on $r(X)$, while that of \cite{Cal69,Hirsch} does not.
Thus each bound is stronger (i.e., smaller) for certain settings of parameters $r$ and $t$.
Roughly speaking, the upper bounds of \cite{Cal69,Hirsch} are stronger than those of \cite{Lev66} for large values of $r$ and $t$; while the opposite is true for small $r$ and $t$.

\begin{figure}[t]
\centering
\includegraphics[scale=0.6]{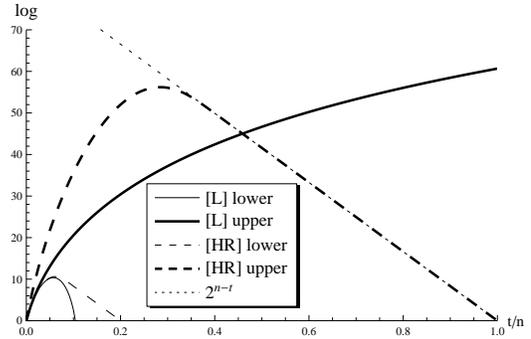}
\caption{Previous bounds on $|D_t(X)|$. [L] marks the bounds proven by Levenstein \cite{Lev66}, and [HR] marks the bounds by Hirschberg et. al \cite{Hirsch}. Also plotted is the naive bound $2^{n-t}$ which is the possible number of binary strings of length $n-t$. This graph shows an example for the case $n=120$, and $r=24$. All graphs are shown on a logarithmic scale.}\label{fig:previous}
\end{figure}

\subsection{Our results and proof techniques}

In this work, we continue the study of $D_t(X)$ and present improved upper and lower bounds to those described above.
Our analysis is two fold.
We start by studying the family of strings $X$ for which $r=r(X)$, and identify the {\em extremal} strings in this family with respect to the number of subsequences.
Specifically, for every $r$, we identify two $r$-run strings, referred to as the {\em balanced} $r$-run string $B_r$ and the {\em unbalanced} $r$-run string $U_r$ such that for every $X$ it holds that 
\begin{equation}
\label{eq:UB}
|D_t(U_{|r(X)|})| \le |D_t(X)| \le |D_t(B_{|r(X)|})|. 
\end{equation}
Loosely speaking, the string $U_r=0101 \dots 01^{n-r+1}$ is the $r$-run string in which each run is exactly of size $1$, except the last run which is of size $n-t+1$, and is thus referred to as `unbalanced' (in the run lengths).
The balanced string $B_r=0^{n/r}1^{n/r}0^{n/r}1^{n/r} \dots 1^{n/r}0^{n/r}$ is the $r$-run string in which each and every run is of equal length $n/r$.

To obtain Equation~(\ref{eq:UB}), we show that any $r$-run string $X$ can be transformed into the string $U_r$ (alternatively $B_r$) via a series of operations that are monotonic with respect to the number of subsequences.
The modifications we study include a {\em balancing} operation, in which given $X$ we shorten the length of one of its runs while increasing the length of another; a {\em flipping} operation, in which a prefix or suffix of $X$ is replaced by it complement; and an {\em insertion} operation in which characters are added to $X$ (see Figures~\ref{fig:insertion},~\ref{fig:flip} and \ref{fig:balance}). 
A delicate combination of these (and other) operations enable us to establish Equation~(\ref{eq:UB}).
The modifications we study and their analysis shed light on the properties of binary strings under the deletion operation and may be of independent interest. 
We note that for the extreme case of $r=n$, our unbalanced string $U_n$ is exactly the cyclic string $C_n$; thus we are consistent with the result of \cite{Cal69}.

\begin{figure}[t] 
\centering
\includegraphics[scale=0.6]{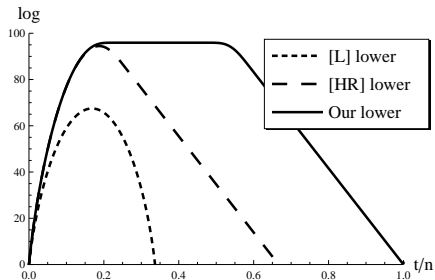}
\caption{Comparision of lower bounds. Our lower bound based on unbalanced strings [Theorem \ref{t:lowerbound}], compared to the previous known bounds. [L] marks the lower bound proven by Levenstein \cite{Lev66}. [HR] marks the lower bound proven by Hirschberg et. al \cite{Hirsch}. This graph shows an example for the case $n=300$, and $r=200$. The logarithmic presentation emphasizes that we obtain an {\em exponential} multiplicative improvement.}
\label{fig:compareBoundsL}
\end{figure}

We then turn to obtain analytic expressions for $|D_t(U_r)|$ and $|D_t(B_r)|$ of Equation~(\ref{eq:UB}).
Our expressions are at least as good as previous bounds in \cite{Lev66,Cal69,Hirsch} as they are based on specific $r$-run strings ($U_r$ and $B_r$), and for a large range of parameters our bounds are strictly tighter.
For our improved lower bound, we devise a recursive expression for $|D_t(U_r)|$ and present a closed form formula for its evaluation.
We then perform an asymptotic evaluation of $|D_t(U_r)|$ (assuming large $r$).
A comparison of our improved lower bound with that previously known is depicted in Figure~\ref{fig:compareBoundsL}.
Specifically, we show that for values of $t$ which are greater than $r/3$ our lower bound improves on those previously known by an exponential multiplicative factor of roughly $2^{t-r/3}$.

To address our improved upper bounds, we first present a recursive formula for the computation of $|D_t(B_r)|$.
We then extract a closed form solution to our recursive definition which yields an exact expression for $|D_t(B_r)|$. 
For example, a numerical comparison of $|D_t(B_r)|$ with the upper bounds previously known is depicted in Figure~\ref{fig:compareBoundsH} for the value of $n=120$ and $r=24$ as a function of $t$.
We note that the expression we obtain for $|D_t(B_r)|$ involves several summations of certain combinatorial expressions.
An asymptotic analysis of our expression is left open in this work and is subject to future research. 

\begin{figure}[t] 
\centering
\includegraphics[scale=0.6]{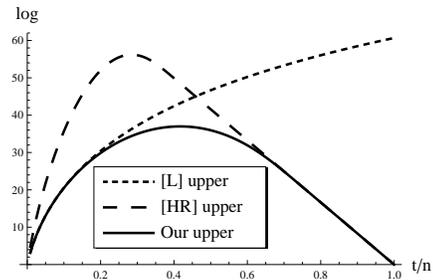}
\caption{Comparison of upper bounds. Our upper bounds based on balanced strings [Corollary \ref{c:subseqB}], compared to the previous best known bounds.  [L] marks the upper bound proven by Levenstein \cite{Lev66}. [HR] marks the upper bound proven by Hirschberg et. al \cite{Hirsch}. This graph shows an example for the case $n=120$ and $r=24$ as a function of $t$ (in logarithmic scale).} \label{fig:compareBoundsH}
\end{figure}

%
 
\subsection{Structure}
The remainder of the paper is organized as follows.
In section~\ref{sec:tools} we present the set of structural operations and tools we use for comparing and bounding the number of subsequences obtained via deletion. This section includes our balancing, flipping, and insertion modifications.
In Section~\ref{sec:balanced}, we study our first family of {\em balanced} strings, and show  that (for any given number of runs $r$ and deletions $t$) they have the largest number of subsequences under deletion.
In Section~\ref{sec:upper}, we analyze the number of subsequences of balanced strings 
and in such obtain our upper bound.
In Section~\ref{sec:unbalanced}, we present our second family of {\em unbalanced} strings, and prove that they have the least number of subsequences under any number of deletions $t$.
We prove our lower bound by analyzing the number of subsequences of unbalanced strings in Section~\ref{sec:lower}. 
Finally, in Section~\ref{sec:dp}, we study the connection between subsequences and the closely related notion of {\em deletion patters}. 
Using this connection, we show exponential multiplicative gaps between our improved upper bound and those previously presented.


\section{Tools for analyzing the number of subsequences}
\label{sec:tools}

The number of subsequences of a string obtained by deletions highly depends on the string's structure. In order to determine the number of subsequences for a given number of deletions, it is not enough to know the length of the string, and not even the number of the string's runs. Inspired by previous works, we looked for tools that will enable us to analyze the number of subsequences. In this section we present these tools. In subsection \ref{ssec:partition} we present a method of counting the number of subsequences by partitioning the set of subsequences into subsets characterized by their prefix, thus forming a recursive relation. In subsection \ref{ssec:operations} we present basic operations on strings that always increase (or decrease) the number of subsequences under deletion. Such basic operations allow comparison between the number of subsequences of strings, and are very useful for finding bounds on the number of subsequences.

$S(x_1,\dots,x_r)$ denotes a binary string with $r$ runs, in which the $i^{th}$ run is of length $x_i$ and the first symbol is $0$, E.g. $S(1,2,3)=011000$. 
We will use the notation $n\times a$ to indicate $n$ sequential runs of length $a$, E.g. $S(2,3\times 1,2)=S(2,1,1,1,2)=0010100$.
$D_t(x_1,\dots,x_r)$ will be used as short form for $D_t(S(x_1,\dots,x_r))$.
$C_n$ denotes the binary cyclic string $S(n\times 1$).
We assume the following conventions: $\sum_{i=j}^k a_i = 0$ when $j>k$. $\binom n i=0$ when $i<0$ or $i>n$. $|D_t(X)=1|$ for $t=|X|$ and $t=0$, and $|D_t(X)=0|$ for $t>|X|$.

\subsection{Partitioning the set of subsequences}\label{ssec:partition}
We found the following lemma (from \cite{Hirsch}) very useful. We restate it here and derive a corollary for binary strings.
\begin{lemma}\label{l:partition}
\cite{Hirsch} For any $\Sigma$-string $X$:
\begin{itemize}
\item[(i)] $D_t(X)=\sum_{a\in \Sigma}{D_t^{(a)}(X)}$, where for a set $G$ of strings $G^{(a)}$ denotes all members of $G$ starting with $a$.
\item[(ii)] $D_t^{(a)}(X)=aD_{t+1-f(a)}(X[f(a)+1:n])$, where $f(a)$ denotes the index of the first appearance of $a$ in $X$, and $X[i:j]$ denotes the substring $x_i\dots x_j$ of $X$.
\end{itemize}
\end{lemma}

We derive the following lemma for binary strings.

\begin{lemma}\label{l:binSplit}
\mbox{}
\begin{itemize}
\item [(i)] For any binary string $X$, s.t. $X=\sigma^i \epsilon^j Y$ for some $i,j>0$ and $Y\in\{\sigma,\epsilon\}^*$, $|D_t(X)|=|D_t(\sigma^{i-1}\epsilon^j Y)|+|D_{t-i}(\epsilon^{j-1} Y)|$ for any $t<|X|$. 
\item [(ii)] Symmetrically, $|D_t(Y\epsilon^j\sigma^i)|=|D_t(Y\epsilon^j\sigma^{i-1})|+|D_{t-i}(Y\epsilon^{j-1})|$.
\end{itemize}
\end{lemma}

\begin{proof}
(i) Following the notation of Lemma \ref{l:partition}, $f(\sigma)=1$ and $f(\epsilon)=i+1$. Using Lemma \ref{l:partition}(ii), $D_t^{(\sigma)}=\sigma D_{t+1-1}(X[2:n])$ and $D_t^{(\epsilon)}=\epsilon D_{t+1-(i+1)}(X[i+2:n])$. Applying Lemma \ref{l:partition}(i) we get the result.

(ii) The proof for the symmetric case is identical.
\end{proof}

Applying Lemma \ref{l:binSplit} repeatedly, we get the following lemma.
\begin{lemma}\label{l:binSplitSum}
For any binary string $S(x_1,\dots,x_r)$, s.t. $n=\sum_{i=1}^r x_i$:
\begin{itemize}
\item[(i)] $|D_t(x_1,\dots,x_r)|=|D_t(x_2,\dots,x_r)|+\sum_{i=1}^{x_1}|D_{t-i}(x_2-1,x_3,\dots,x_r)|+1|_{t>n-x_1}$.
\item[(ii)] Symmetrically, $|D_t(x_1,\dots,x_r)|=|D_t(x_1,\dots,x_{r-1})|+\sum_{i=1}^{x_r}|D_{t-i}(x_1,\dots,x_{r-2},x_{r-1}-1)|+1|_{t>n-x_r}$.
\end{itemize}
\end{lemma}
\begin{proof}
(i) We denote $n=\sum_{i=1}^r x_i$. 
Using Lemma \ref{l:binSplit} once, we get $|D_t(x_1,\dots,x_r)|=|D_t(x_1-1,x_2,\dots,x_r)|+|D_{t-x_1}(x_2-1,x_3,\dots,x_r)|$. For $x_1>1$ we can use Lemma \ref{l:binSplit} again and get  $|D_t(x_1,\dots,x_r)|=|D_t(x_1-2,x_2,\dots,x_r)|+|D_{t-x_1+1}(x_2-1,x_3,\dots,x_r)|+|D_{t-x_1}(x_2-1,x_3,\dots,x_r)|$. Likewise, for $j\le \min(x_1,n-t)$, applying Lemma \ref{l:binSplit} $j$ times yields $|D_t(x_1,\dots,x_r)|=|D_t(x_1-j,x_2,\dots,x_r)|+\sum_{i=x_1-j+1}^{x_1}|D_{t-i}(x_2-1,x_3,\dots,x_r)|$. 
When $t\le n-x_1$, it follows that $\min(x_1,n-t)=x_1$, and so we can expand using Lemma \ref{l:binSplit} exactly $x_1$ times to get  $|D_t(x_1,\dots,x_r)|=|D_t(x_2,\dots,x_r)|+\sum_{i=1}^{x_1}|D_{t-i}(x_2-1,x_3,\dots,x_r)|$. 
When $t>n-x_1$ after expanding $n-t$ times, we get the expression $|D_t(x_1,\dots,x_r)|=|D_t(x_1-(n-t),x_2,\dots,x_r)|+\sum_{i=x_1-(n-t)+1}^{x_1}|D_{t-i}(x_2-1,x_3,\dots,x_r)|$. 
As $|S(x_1-(n-t),x_2,\dots,x_r)|=t$, and noticing that for $t>|X|$, $|D_t(X)|=0$, we get the lemma's claim.

(ii) The proof for the symmetric case is identical.
\end{proof}

\subsection{Basic operations on strings}\label{ssec:operations}
\begin{figure}[t]
  \begin{center}
   \subfigure[Insertion]{\label{fig:insertion}\includegraphics[scale=0.55]{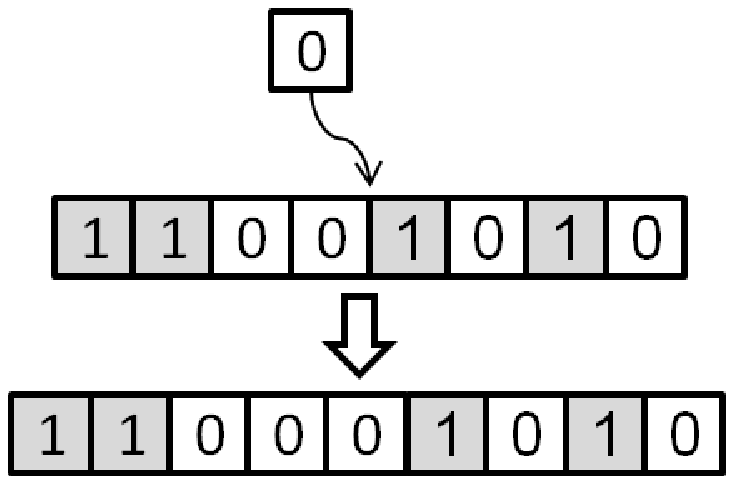}}
   \subfigure[Flip]{\label{fig:flip}\includegraphics[scale=0.52]{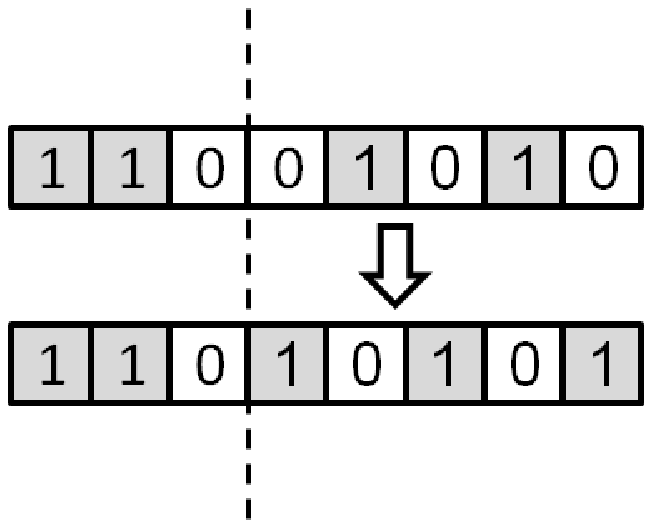}} 
   \subfigure[Balance]{\label{fig:balance}\includegraphics[scale=0.55]{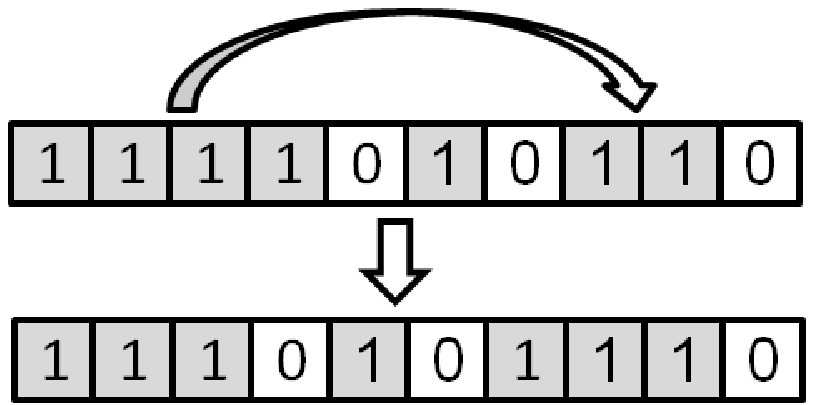}}
  \end{center}
  \caption{Basic operations on strings. In all diagrams the lower string has more subsequences under any number of deletions.}
  \label{fig:basicOp}
\end{figure}
In the following sections we will present families of strings, for which the number of subsequences can be explicitly calculated. In order to use these families of strings to devise bounds on the number of subsequences for general strings, we use basic operations on strings, which allows us to transform one string into another, while monotonically increasing (or decreasing) the number of their subsequences. In this section we list such basic operations.

\subsubsection {Insertion operation [Figure \ref{fig:insertion}]}
Hirschberg et al.\cite{Hirsch} showed that inserting a symbol anywhere in the middle of a string always increases the number of subsequences. 
\begin{lemma}\label{l:insertion}[Insertion increases the number of subsequences]
\cite{Hirsch} For any $\Sigma$-strings $U$, $V$ and any $\sigma \in \Sigma$, $|D_t(UV)| \le |D_t(U\sigma V)|$. 
\end{lemma}

\subsubsection {Deletion chain rule}
\begin{lemma}\label{l:associativity}
For any  $\Sigma$-string $U$, and any $V\in D_t(U)$, $D_{t'}(V)\subseteq D_{t+t'}(U)$.
\end{lemma}
\begin{proof}
$V$ was obtained from $U$ by deleting $t$ symbols. Any string in $D_{t'}(V)$ is obtained by deleting $t'$ symbols from $V$. The same string can be created by removing the $t+t'$ symbols directly from $U$, and thus it belongs also to $D_{t+t'}(U)$
\end{proof}   

\subsubsection {Flipping operation}[Figure \ref{fig:flip}]
\begin{lemma}\label{l:flip}[Flipping increases number of subsequences]
For any binary strings $U$, $V$ and for any bit $\sigma$, $|D_t(U \sigma \sigma V)| \le |D_t(U\sigma \overline{\sigma V})|$, where $\overline a$ denotes the string $a$ in which 0's are flipped to 1's, and vice versa.
\end{lemma}
\begin{proof}
By induction on $|U|$. When $|U|=0$ the claim is $|D_t(\sigma \sigma V)| \le |D_t(\sigma \overline{\sigma V})|$. Let $V=\sigma^i \epsilon^j X$ for maximal $i$,$j$. When $j=0$ the claim is trivial ($\sigma \sigma V$ is a constant string, with 1 possible subsequence), so we assume $j>0$. Using Lemma \ref{l:binSplit} we get $|D_t(\sigma\sigma V)|=|D_t(\sigma V)|+|D_{t-2-i}(\epsilon^{j-1}X)|$. We compare that to $|D_t(\sigma\epsilon \overline V) = |D_t(\epsilon\overline V)| + |D_{t-1}(\overline V)|$. $|D_t(\sigma V)|=|D_t(\epsilon\overline V)|$ because of symmetry, and since $\epsilon^{j-1}X\subseteq D_{i+1}(V)$ we can use Lemma \ref{l:associativity} and get $|D_{t-2-i}(\epsilon^{j-1}X)| \le |D_{t-1}(\overline V)|$, and thus we prove the base of the induction.

\vspace{2mm} \noindent
Now for the induction step, assume the claim is true for $|U|<n$ and we look at $|U|=n$. We regard the different cases of the structure of $U$.
\begin{itemizei}
\item [Case 1:]  $U=\sigma^i\epsilon^jX$ for some $i,j>0$. We use Lemma \ref{l:binSplit} and get $|D_t(\sigma^i\epsilon^j X\sigma\sigma V)|=| D_t(\sigma^{i-1}\epsilon^j X\sigma\sigma V)| + |D_{t-i}(\epsilon^{j-1} X\sigma\sigma V)|$. We compare that to $|D_t(\sigma^i\epsilon^j X\sigma\epsilon \overline V)|=|D_t(\sigma^{i-1}\epsilon^j X\sigma\epsilon \overline V)| + |D_{t-i}(\epsilon^{j-1}X\sigma\epsilon \overline V)|$. On each of the arguments we can use our induction claim for $|U|-1$ and $|U|-i-1$.
\item [Case 2:] $U=\epsilon^i\sigma^jX$ for some $i,j>0$ we use the same method.
\item [Case 3:] $U=\epsilon^i$ for some $i>0$. $|D_t(\epsilon^i \sigma\sigma V)|=|D_t(\epsilon^{i-1}\sigma\sigma V)|+|D_{t-i}(\sigma V)|$. For the flipped string we get $|D_t(\epsilon^i \sigma\epsilon \overline V)|=|D_t(\epsilon^{i-1}\sigma\epsilon \overline V)|+|D_{t-i}(\epsilon \overline V)|$. In this case, the second argument in both summations is equal due to symmetry, and we can compare the first arguments using the induction hypothesis for $|U|-1$.
\item [Case 4:] $U=\sigma^i$ for some $i>0$. Let $V=\sigma^j X$ for maximal $j$. In case $|X|=0$ we get the trivial case of a uniform string again. For $|X|>0$ let $X=\epsilon Y$, and then $|D_t(U\sigma\sigma V)|=|D_t(\sigma^{i+2}V)| = |D_t(\sigma^{i+1}V)|+|D_{t-i-j-2}(Y)|$. Again we compare that to $|D_t(\sigma^{i+1}\epsilon \overline V)| = |D_t(\sigma^i \epsilon \overline V)|+|D_{t-i-1}(\overline V)|$, using the induction claim for the first argument, and Lemma \ref{l:associativity} together with symmetry for the second.
\end{itemizei}
\end{proof}

\begin{corollary} [Alternative proof for the maximality of $C_n$]
Given any string $X$ of length $n$, it can be transformed into the string $C_n$ by a series of flipping operation (as defined in Lemma \ref{l:flip}). Each such flip can only increase the number of subsequences, and thus we get a proof for the fact that $D_t(X)\le D_t(C_n)$. 
\end{corollary}
\mbox{}

\subsubsection{Balancing operation}[Figure \ref{fig:balance}]
Informally, we refer to a string as {\em balanced}, if there is a low variability between the length of the string's runs. A balancing operation is one that decreases that variability, E.g. shortening a long run and increasing the length of a short run. The following lemma states terms in which balancing a string increases the number of its subsequences, and it is used later to prove maximality of string families.
\begin{lemma}\label{l:balanceOp}[Balancing increases the number of subsequences] For $X=S(x_1,\dots,x_r)$, and for any $t>0$, $1\le i<j\le r$ s.t. $x_i-x_j>1$, and $\{x_{i+1},\dots,x_{j-1}\}$ is symmetric  (i.e. $ x_2=x_{r-1}$, $x_3=x_{r-2}$,$\dots$), $|D_t(x_1,\dots,x_r)|\le |D_t(x_1,\dots,x_{i-1},x_i-1,x_{i+1},\dots,x_{j-1},x_j+1,x_{j+1},\dots,x_r)|$. In other words, decreasing the i-th run by 1, and increasing the j-th run by 1 can only increase the number of subsequences. 
\end{lemma}

In order to prove Lemma \ref {l:balanceOp} we will need the following lemma that characterizes balancing operations near the edges of the string.
\begin{lemma}\label{l:balanceEdge}
Assume $\{x_2,\dots,x_{r-1}\}$ is symmetric. It follows that:
\begin {itemize}
\item [(i)] For $X=S(x_1,\dots,x_r)$ s.t. $x_1>x_r$, $|D_t(x_1,\dots\,x_r)|\le |D_t(x_1-1,x_2,\dots,x_{r-1},x_r+1)|$.  
\item [(ii)] For $X=S(x_1,\dots,x_r,z)$ s.t. $x_1>x_r$ and $z>0$, $|D_t(x_1,\dots\,x_r,z)|\le |D_t(x_1-1,x_2,\dots,x_{r-1},x_r+1,z)|$.  
\item [(iii)] For $X=S(y,x_1,\dots,x_r)$ s.t. $x_1-x_r>1$ and $y>0$, $|D_t(y,x_1,\dots\,x_r)|\le |D_t(y,x_1-1,x_2,\dots,x_{r-1},x_r+1)|$. 
\item [(iv)] For $X=S(y,x_1,\dots,x_r,z)$ s.t. $x_1-x_r>1$ and $y,z>0$, $|D_t(y,x_1,\dots\,x_r,z)|\le |D_t(y,x_1-1,x_2,\dots,x_{r-1},x_r+1,z)|$.  
\end{itemize}
\end{lemma}

\begin{proof}
\begin {itemize}
\item [(i)] When $r=2$ the claim is reduced to $D_t(x_1,x_2)\le D_t(x_1-1,x_2+1)$ for $x_1>x_2$. This is easily proved because $D_t(x_1,x_2)=\min(x_1,x_2,t)+1$. 
For $r>2$ we use Lemma \ref{l:binSplit} to get $|D_t(x_1,\dots\,x_r)|=|D_t(x_1-1,x_2,\dots,x_r)|+|D_{t-x_1}(x_2-1,x_3,\dots,x_r)|$. 
Using Lemma \ref{l:binSplit} and the symmetry of $\{x_2,\dots,x_{r-1}\}$ we get $|D_t(x_1-1,x_2,\dots,x_{r-1},x_r+1)|=|D_t(x_r+1,x_2,\dots,x_{r-1},x_1-1)|=|D_t(x_r,x_2,\dots,x_{r-1},x_1-1)|+|D_{t-x_r-1}(x_2-1,x_3,\dots,x_{r-1},x_1-1)|$. 
We compare the two expressions. Because of the symmetry $|D_t(x_1-1,x_2,\dots\,x_r)|=|D_t(x_r,x_2,\dots,x_{r-1},x_1-1)|$, and because $x_1>x_r$ it is true that $S(x_2-1,x_3,\dots,x_r)\in D_{x_1-x_r-1}(x_2-1,x_3,\dots,x_1-1)$ and thus using Lemma \ref{l:associativity} $|D_{t-x_1}(x_2-1,x_3,\dots,x_r)|\le |D_{t-x_r-1}(x_2-1,x_3,\dots,x_1-1)|$.
\item [(ii)] Applying Lemma \ref{l:binSplitSum}(ii) we get $|D_t(x_1,\dots\,x_r,z)|=|D_t(x_1,\dots\,x_r)|+\sum_{i=1}^z|D_{t-i}(x_1\dots\,x_{r-1},x_r-1)|+1|_{t>n-z}$, and $|D_t(x_1-1,x_2,\dots,x_{r-1},x_r+1,z)|= |D_t(x_1-1,x_2,\dots\,x_{r-1},x_r+1)|+\sum_{i=1}^z|D_{t-i}(x_1-1,x_2,\dots\,x_r)|+1|_{t>n-z}$. The two expressions are comparable argument by argument using (i) above, noticing that if $x_1>x_r$ then definitely $x_1>x_r-1$.
\item [(iii)] Applying Lemma \ref{l:binSplitSum} we get $|D_t(y,x_1,\dots\,x_r)|=|D_t(x_1,\dots\,x_r)|+\sum_{i=1}^y|D_{t-i}(x_1-1,x_2,\dots\,x_r)|+1|_{t>n-y}$, and $|D_t(y,x_1-1,x_2,\dots,x_{r-1},x_r+1)|= |D_t(x_1-1,x_2,\dots\,x_{r-1},x_r+1)|+\sum_{i=1}^y|D_{t-i}(x_1-2,x_2,\dots\,x_{r-1},x_r+1)|+1|_{t>n-y}$. The two expressions are comparable argument by argument using (i) above and noticing that if $x_1-x_r>1$ then definitely $x_1>x_r$ and $x_1-1>x_r$.
\item [(iv)] We use Lemma \ref{l:binSplitSum} to get $|D_t(y,x_1,\dots\,x_r,z)|=|D_t(x_1,\dots\,x_r,z)|+\sum_{i=1}^y|D_{t-i}(x_1-1,x_2,\dots\,x_r,z)|+1|_{t>n-y}$ and $|D_t(y,x_1-1,x_2,\dots,x_{r-1},x_r+1,z)|=|D_t(x_1-1,x_2,\dots,x_{r-1},x_r+1,z)|+\sum_{i=1}^y|D_{t-i}(x_1-2,x_2,\dots,x_{r-1},x_r+1,z)|+1|_{t>n-y}$. The two expressions are comparable argument by argument using (ii) above, as the condition $x_1-x_r>1$ guarantees that $x_1-1>x_r$. 
\end{itemize}
\end{proof}
Now we can prove Lemma \ref {l:balanceOp} [Balancing increases the number of subsequences]:

\begin{proof}
We will prove by induction on the number of runs in $X$ outside of the sequence $S(x_i,\dots,x_j)$, explicitly on $(i-1)+(r-j)=r+i-j-1$. We will denote these runs {\em outer runs}. When we have only one outer run, the lemma is reduced to Lemma \ref{l:balanceEdge}(ii) or \ref{l:balanceEdge}(iii). Now we assume that there are at least two outer runs. If the outer runs are one on each side ($i=2$ and $j=r-1$) this is the case of Lemma \ref{l:balanceEdge}(iv). Otherwise, at least on one of the sides there are two or more runs ($i>2$ or $j<r-1$). We assume w.l.o.g that $i>2$, and then we can use Lemma \ref{l:binSplitSum} and the induction hypothesis on strings with the number of outer runs decreased by 1. 

\end{proof}

\section{Balanced strings}
\label{sec:balanced}
In this section we define the family of strings named {\em Balanced strings}. We call a string balanced, if all the runs of symbols in the string are of equal length. Formally, we denote by $B_{r,k}$ the binary string of length $rk$, with $r$ runs, each of length $k$. E.g. $B_{3,4}=S(4,4,4)=000011110000$. We will prove that of all strings with length $rk$ and $r$ runs, the balanced string has the maximal number of subsequences, under any number of deletions.

\begin{theorem}\label{t:balancedMax}
Let $X=S(x_1,\dots,x_r)$, $n=\sum_{i=1}^r x_i$, and $k=n/r$. If $k$ is an integer, then $|D_t(X)|\le |D_t(B_{r,k}) |$.
\end{theorem}

\begin{proof}
The main idea of the proof is that any such string $X$ can be transformed into $B_{r,k}$ by repeatably applying the Balancing Lemma \ref{l:balanceOp}. Each such step can only increase the number of subsequences, so if such a series of balance operations can be found, the theorem is proved. We will construct a series of strings, $X_0,\dots,X_m$, such that $X_0=X$, $X_m=B_{r,k}$ and for any $0\le i<m$, $|D_t(X_i)|\le|D_t(X_{i+1})|$. Given a string $X_i\neq B_{r,k}$, we denote $X_i=S(x^{(i)}_1,\dots,x^{(i)}_r)$. We choose a pair $(p,q)$ s.t $|x^{(i)}_p-x^{(i)}_q|>1$, $p<q$ and $q-p$ is minimal. Such a pair exists, because at least one run is of length different from $k$ (w.l.o.g, bigger than $k$), and thus there is at least one other run with length smaller than $k$. Assume w.l.o.g that $x^{(i)}_p>x^{(i)}_q$, and then we can conclude that $x^{(i)}_p> x^{(i)}_{p+1}=x^{(i)}_{p+2}=\dots=x^{(i)}_{q-1}> x^{(i)}_q$, otherwise we get a contradiction to the minimality of $(p,q)$. We will define $X_{i+1}$ to be the string achieved from $X_i$ by decreasing the $p^{th}$ run by 1, and increasing the $q^{th}$ run by 1. Each pair of strings $X_i,X_i+1$ admits to the conditions of Lemma \ref{l:balanceOp} and thus $|D_t(X_i)|\le|D_t(X_{i+1})|$. This process is finite, because the value of $\sum_{i=0}^r x_i^2$ is a non negative integer that must decreases at every step. An example of the balancing process we use is displayed in Table I.
\end{proof}

\begin{table}
\caption{Example of a balancing process as defined in the proof of Theorem \ref{t:balancedMax}}
\begin{center}
\begin{tabular}{ l c c c c}
$i$ & $X_{i}$ & runs & $\sum x_i^2$ & $D_6(X_i)$  \\ 
0 & 000111111100100 & {3,7,2,1,2} & 67 & 43\\
1 & 000111111000100 & {3,6,3,1,2} & 59 & 56\\
2 & 000111110000100 & {3,5,4,1,2} & 55 & 63\\
3 & 000111110001100 & {3,5,3,2,2} & 51 & 85\\
4 & 000111100001100 & {3,4,4,2,2} & 49 & 92\\
5 & 000111100011100 & {3,4,3,3,2} & 47 & 102\\
6 & 000111000111000 & {3,3,3,3,3} & 45 & 105\\
\end{tabular}
\end{center}
\end{table}

We derive the following corollary for the case where $n$ is not divisible by $r$.
\begin{corollary}\label{c:balancedMax}
Let $X=S(x_1,\dots,x_r)$, $n=\sum_{i=1}^r x_i$, and $\bar{k}=n/r$.  $D_t(X)\le |D_t(B_{r,\lceil \bar{k} \rceil}) |$.
\end{corollary}
\begin{proof}
For integral $\bar{k}$ this is the case of Theorem \ref{t:balancedMax}. Otherwise, we denote $\alpha=r\lceil \bar{k} \rceil -n$, and let $Y=|D_t(x_1,\dots,x_{r-1},x_r+\alpha)|$. Using Lemma \ref{l:insertion} $|D_t(X)|\le|D_t(Y)|$, and since  $|Y|= r\lceil \bar{k} \rceil $ and $r(Y)=r$,  using Theorem \ref{t:balancedMax} $|D_t(Y)\le |D_t(B_{r,\lceil \bar{k} \rceil})|$.
\end{proof}

\section{Our Upper bound}
\label{sec:upper}

In this section we present an upper bound for the number of subsequences of a string obtained by deletions. 
We develop a recursive expression for the exact number of subsequences of a balanced string. We then find an explicit form for this expression, and use it to obtain a tight upper bound on the number of subsequences of a general string. 

\subsection{Recursive expression}
\begin{definition}
For all $r,k$ , Let $B'_{r,k}$ be the string obtained from $B_{r,k}$ by removing the first symbol. E.g. $B\rq_{3,5}=S(4,5,5)=00001111100000$.
\end{definition}

\begin{definition}
Let  $b(r,k,t)=|D_t(B_{r,k})|$ and $b'(r,k,t)=|D_t(B'_{r,k})|$.
\end{definition}

\begin{lemma}\label{l:BSplit}
For all $r,k,t$, $|D_t(B_{r,k})| = |D_t(B'_{r,k})|+|D_{t-k}(B'_{r-1,k})|$
\end{lemma}
\begin{proof}
This is derived from Lemma \ref{l:binSplit}
\end{proof}

When $k$ is known from the context, we will use the short notations $B_r$ and $B'_r$ for $B_{r,k}$ and $B'_{r,k}$ repectively. Likewise $b(r,t)$ and $b'(r,t)$ denote $b(r,k,t)$ and $b'(r,k,t)$ respectively. 

\begin{lemma}[Recursive expression for $b\rq$]
\label{l:recFormB}
\[
  b\rq(r,t)=
  \begin{cases}
     0  &\text{if}\; t<0 \;\text{or}\; t\ge kr\\
     1+\sum_{i=1}^{k-1}b\rq(r-1,t-i)  &\text{if}\; k(r-1)\le t<kr \\
     \\[1ex]
     \begin{aligned}
     b\rq(&r-2,t-k)+  \\
	&\sum_{i=0}^{k-1}b\rq(r-1,t-i) 
     \end{aligned}  &\text{otherwise}
  \end{cases}
\]
\end{lemma}
\begin{proof}
Using Lemma \ref{l:binSplitSum} we get $b'(r,t)=b(r-1,t)+\sum_{i=1}^{k-1}b'(r-1,t-i)+1|_{t>k(r-1)}$. We check the following cases:
\begin{itemizei}
\item[(i) $t<k(r-1)$:] Using Lemma \ref{l:binSplit}, $b(r-1,t)=b'(r-1,t)+b'(r-2,t-k)$, and we get $b'(r,t)=b'(r-2,t-k)+\sum_{i=0}^{k-1}b'(r-1,t-i)$.
\item[(ii) $t=k(r-1)$:] In this case $t=|B_{r-1}|$ and $b(r-1,t)=1$. We get $b'(r,t)=1+\sum_{i=1}^{k-1}b'(r-1,t-i)$.
\item[(iii) $t>k(r-1)$:] Here $t>|B_{r-1}|$ and $b(r-1,t)=0$. We get $b'(r,t)=\sum_{i=1}^{k-1}b'(r-1,t-i)+1$.
\end{itemizei}
Rearranging the cases we get the claim of the lemma.
\end{proof}
\subsection{Solving the recursion}

When calculating $b\rq(r,t)$ we expand the recursive expression iteratively, until all $b\rq$ expressions reach their boundary condition, and get zero value. The only positive contribution in this sum is from the 1 in the second case ( $1+\sum_{i=1}^{k-1}b\rq(r-1,t-i)$ ). By counting how many times this value is added, we can get the explicit value of $b\rq(r,t)$. The 1 values are added exactly every time the second case is used, i.e. when expanding the value of $b\rq(\tilde r,\tilde t)$ for $\tilde r,\tilde t$ that fulfill the condition $k(\tilde r-1)\le \tilde t<k\tilde r$. When expanding $b\rq(r,t)$ these are exactly the integral solutions for $\tilde r=\lfloor\frac {\tilde t} k\rfloor+1, 0\le \tilde t \le t$, which are simply the $t+1$ pairs $(r_i,t_i)=(\lfloor\frac{i}{k}+1\rfloor,i)$ for $0\le i \le t$.
We will count the number of times that $b\rq(\tilde{r},\tilde{t})$ appears in the complete expansion of $b\rq(r,t)$. Based on the recursion form in Lemma \ref{l:recFormB}, the expression $b\rq(\tilde{r},\tilde{t})$ can only appear in the single expansion of one of the following expressions: $b\rq(\tilde{r}+2,\tilde{t}+k)$, or $b\rq(\tilde{r}+1,\tilde{t}+i)$ when $0\le i \le k-1$. Counting the number of those paths is equivalent to calculating the number of possible sets of ordered tuples $\{(r_j,t_j)\}$ selected from the set $\{(2,k),(1,0),(1,1),\dots,(1,k-1)\}$ s.t. $\sum r_j = r-\tilde{r}$ and $\sum t_j = t-\tilde{t}$.

\begin{definition}
We denote as $\mathbb S_k$ the set $\{(2,k),(1,0),(1,1),\dots,(1,k-1)\}$, and as $\#P(\Delta r, \Delta t)$ the number of possible sets of ordered tuples $\{(r_j,t_j)\}$ selected from the set $\mathbb S_k$ s.t. $\sum r_j = \Delta r$ and $\sum t_j = \Delta t$. 
$\#P_j(\Delta r, \Delta t)$ will denote the number of such sets using the tuple $(2,k)$ exactly $j$ times.
\end{definition}

\begin{lemma}\label{l:calcP0}
$$\#P_0(\Delta r, \Delta t)=\sum_{i=0}^{\lfloor \frac{\Delta t}k \rfloor}(-1)^i\binom{\Delta r}{i}\multiset{\Delta r}{\Delta t-ik}$$
\end{lemma}
\begin{proof}
In the case of $\#P_0$, the problem is reduced to finding the number of ordered partitions of $t$ into $r$ parts, each of size no larger than $k-1$. 
The following development follows the technique used is \cite{Bhar} in the context of counting {\em deletion patterns}, similar results are calculated in \cite{Ratsaby}. 
This partitioning problem can be restated as counting the different solutions $\{y_i\}$ to the equations $\sum_{i=0}^r y_i=t$, $\forall i: y_i<k$.
The number of solutions ignoring the constraints $y_i< k$ is equivalent to the number of r-partitions of $t$, which is $\multiset{r}{t}=\binom{r+t-1}{t}$. 
The number of solutions that violate the constraint $y_1< k$ is $\multiset{r}{t-k}$. Subtracting this for each $y_i$, we get $\multiset{r}{t}-r\multiset{r}{t-k}$. Now we subtracted too much, because solutions that violate two constraints are subtracted twice. 
The number of solutions that violate the two constraints $y_1<k$ and $y_2<k$ is $\multiset{r}{t-2k}$, and there are $\binom{r}2$ such pairs. 
Adding these cases back to the count we get $\multiset{r}{t}-r\multiset{r}{t-k}+\binom{r}2\multiset{r}{t-2k}$ 
Now again we have to account for the solutions that violate 3 constraints, that were added too many times, and so on. 
Putting it all together we get $\#P_0(r, t)=\sum_{i=0}^{\lfloor \frac t k \rfloor}(-1)^i\binom{r}{i}\multiset{r}{t-ik}$.

\end{proof}

\begin{lemma}\label{l:calcP}
$$\#P(\Delta r,\Delta t)=\sum_{j=0}^{\lfloor \frac{\Delta t}k \rfloor}\binom{\Delta r-j}{j}\#P_0(\Delta r-2j,\Delta t-jk)$$
\end{lemma}
\begin{proof}
First we calculate $\#P_j(\Delta r,\Delta t)$. If we first select $j$ times the tuple $(2,k)$,we are left with $\#P_0(\Delta r-2j,\Delta t-jk)$ ways to select the remaining tuples. We than have $\binom{\Delta r-j}{j}$ ways to insert the $(2,k)$ tuples inside the rest, and thus $\#P_j(\Delta r,\Delta t)=\binom{\Delta r-j}{j}\#P_0(\Delta r-2j,\Delta t-jk)$.
Summing on all possible $j$-s, $\#P(\Delta r,\Delta t)=\sum_{j=0}^{\lfloor \frac{\Delta t}k \rfloor}\#P_j(\Delta r,\Delta t)$ and the lemma's claim follows.
\end{proof}

\begin{lemma} \label{l:calcBprime}
$$b\rq(r,t)=\sum_{i=0}^t\#P(r-\lfloor\frac{i}{k}\rfloor-1,t-i)$$
\end{lemma}
\begin{proof}
As mentioned in the discussion above, when expanding $b'(r,t)$, Exactly $t+1$ pairs $(\tilde r,\tilde k)$ are reached that fulfill the conditions  $ k(\tilde r-1)\le \tilde t<k\tilde r$, $0\le \tilde t \le t$ and thus contribute to the sum. These are exactly the $t+1$ pairs $(r_i,t_i)=(\lfloor\frac{i}{k}+1\rfloor,i)$ for $0\le i \le t$, and each one of them is reached $\#P(r-r_i,t-t_i)$ times. Summing all together we get $b'(r,t)=\sum_{i=0}^t\#P(r-\lfloor\frac{i}{k}+1\rfloor,t-i)$ which is equal to the lemma's claim.
\end{proof}

\begin{corollary}\label{c:subseqB}
The combined results of Lemmas \ref{l:BSplit}, \ref{l:calcBprime}, \ref{l:calcP} and \ref{l:calcP0} give an explicit expression for $|D_t(B_{r,k})|$.
We restate the results here: 

$b(r,t)=b'(r,t)+b'(r-1,t-k)$

$b'(r,t)=\sum_{i=0}^t\#P(r-\lfloor\frac{i}{k}\rfloor-1,t-i)$

$\#P(\Delta r,\Delta t)=\sum_{j=0}^{\lfloor \frac{\Delta t}k \rfloor}\binom{\Delta r-j}{j}\#P_0(\Delta r-2j,\Delta t-jk)$

$\#P_0(\Delta r, \Delta t)=\sum_{i=0}^{\lfloor \frac{\Delta t}k \rfloor}(-1)^i\binom{\Delta r}{i}\multiset{\Delta r}{\Delta t-ik}$
\end{corollary}


Using balanced strings we have achieved upper bounds for the number of subsequences of general strings. 
Our bound of Corollary \ref{c:subseqB} (in comparison to previous bounds) is depicted in Figure \ref{fig:compareBoundsH}.

\section{Unbalanced strings}
\label{sec:unbalanced}
In the section we define a second family of strings, named {\em unbalanced strings}. We call a string unbalanced, if all of the runs of symbols in the string are of length 1, except for one run. Let $U_{n,r}^{(i)}$ be a binary string of length $n$ with $r$ runs, in which all runs are of length 1, except for the $i^{th}$ run which is of length $n-r+1$. We notice that due to symmetry $|D_t(U^{(1)}_{n,r})|=|D_t(U^{(r)}_{n,r})|$, and define $u(n,r,t)=|D_t(U^{(1)}_{n,r})|=|D_t(U^{(r)}_{n,r})|$. We will show that these extreme cases have the least number of subsequences among the unbalanced strings, and conclude that they have the least amount of subsequences among all strings.
\begin{theorem}\label{t:lower} [Unbalanced strings have the least subsequences]
For $X=S(x_1,\dots,x_r)$, $n=\sum_{i=1}^r x_i$, and any $1\le t \le n$, $|D_t(X)|\ge u(n,r,t)$.
\end{theorem}
\begin{proof}
First we will prove that there exists $j$ s.t.   $|D_t(X)|\ge |D_t(U^{(j)}_{n,r})|$, for all $t$. We notice that the balancing operation of Lemma \ref{l:balanceOp} can be used in the other direction, as an unbalancing operation. We will transform the string $X$ into a string $U^{(j)}_{n,r}$ by repeatably applying the unbalancing operation. each such step can only decrease the number of subsequences, so by constructing a series of such operations, we will prove that $D_t(X)\ge |D_t(U^{(j)}_{n,r})|$. Let $j$ be the index of a maximal run in $X$. We will construct a series of strings, $X_0,\dots,X_m$, such that $X_0=X$, $X_m=U^{(j)}_{n,r}$ and for any $0\le i<m$, $|D_t(X_i)|\ge|D_t(X_{i+1})|$. For any $i<m$, we denote $X_i=S(x^{(i)}_1,\dots,x^{(i)}_r)$. We choose an index $p\ne j$ s.t. $x^{(i)}_p>1$ and all runs between the $j^{th}$ run and the $p^{th}$ run are all of length 1. Such an index exists, otherwise $X_i$ is already an unbalanced string. We define $X_{i+1}$ to be the string obtained from $X_i$ by increasing the $j^{th}$ run by 1, and decreasing the $p^{th}$ run by 1. Since $x_j$ was the maximal run in $X$ and each operation only made it bigger while all other runs could only shorten, we have that $x^{(i)}_j\ge x^{(i)}_p$. The runs between the  $j^{th}$ run and the $p^{th}$ run are all of length 1, and so trivially symmetric, and so the conditions of the reverse Lemma  \ref{l:balanceOp} holds, and $|D_t(X_i)|\ge|D_t(X_{i+1})|$.
 
To complete the proof we will prove that for any $j$, $|D_t(U^{(j)}_{n,r})|\ge u(n,r,t)$. For $j=1$, $u(n,r,t)=|D_t(U^{(1)}_{n,r})|$ by definition. For $j\ge 2$ we will prove by induction on $j$. For $j=2$,  $|D_t(U^{(2)}_{n,r})|=|D_t(1,n-r+1,(n-2)\times 1)|$. using Lemma \ref{l:binSplit} we get  $|D_t(U^{(2)}_{n,r})|= |D_t(n-r+1,(n-2)\times 1)|+ |D_{t-1}(n-r,(n-2)\times 1)|$. We compare this to $u(n,r,t)=|D_t((n-1)\times 1,n-r+1)|=|D_t((n-2)\times 1,n-r+1)|+|D_t((n-3)\times 1,n-r+1)|$. Using the flipping Lemma \ref{l:flip} on the second addend and symmetry on both, we get $u(n,r,t)\le|D_t(n-r+1,(n-2)\times 1)|+|D_t(n-r,(n-2)\times 1)|=|D_t(U^{(2)}_{n,r})|$. 
For the induction step, we assume that the claim is true for $2,\dots,j-1$ and prove it for $j$. for $j>2$, using Lemma \ref{l:binSplit} we get $|D_t(U^{(j)}_{n,r})|=|D_t(U^{(j-1)}_{n-1,r-1})|+|D_{t-1}(U^{(j-2)}_{n-2,r-2})|$. Using the induction assumption on both addends, we get that $|D_t(U^{(j)}_{n,r})|\ge |u(n-1,r-1,t)+u(n-2,r-2,t)$ and using Lemma \ref{l:binSplit} again, the last sum is equal to $u(n,r,t)$ and thus the induction step is proved. An example of the unbalancing process is displayed in Table II.
\end{proof}

\begin{table}
\caption{Example of a balancing process as defined in the proof of Theorem \ref{t:lower}}
\begin{center}
\begin{tabular}{ l c c c}
$i$ & $X_{i}$ & runs & $D_5(X_i)$  \\ 
0 & 0011100111100 & {2,3,2,4,2} & 60\\
1 & 0011101111100 & {2,3,1,5,2} & 38\\
2 & 0011101111110 & {2,3,1,6,1} & 26\\
3 & 0011011111110 & {2,2,1,7,1} & 20\\
4 & 0010111111110 & {2,1,1,8,1} & 14\\
5 & 0101111111110 & {1,1,1,9,1} & 10\\ \hline
   & 1111111110101 & {9,1,1,1,1} & 8\\
\end{tabular}
\end{center}
\end{table}

\section{Our lower bound}
\label{sec:lower}

In this section we develop a recursive expression for the number of subsequences of an unbalanced string by deletions. We will find an explicit form for this expression, and use it to obtain a lower bound on the number of subsequences of a general string. In addition, we will show the improvement that our lower bound provides.

\subsection{Recursive expression}

\begin{lemma}\label{l:URecursion}
For all $0<r\le n$, $0<t<n$,
\[
  u(n,r,t)=
  \begin{cases}
     r & \text{if}\; r=1,2 \\
     2 & \text{if} \; r>1 \text{ and } t=n-1 \\
     d(n,t) & \text{if}\; n=r \\
     \\[1ex]
     \begin{aligned}
     u&(n-1,r,t)+ \\
     &d(r-2,t+r-n-1) 
     \end{aligned} & \text{otherwise}
  \end{cases}
\]

Where $d(r,t)=|D_t(C_r)|=\sum_{i=0}^t\binom{r-t}i$, as proved in \cite{Hirsch}.
We assume $d(n,0)=1$, and for $t<0$, $d(n,t)=0$.
\end{lemma}

\begin{proof}
\begin{itemize}
\item When $r=1$, $U_{n,r}$ is a constant string, and has only one possible subsequence (the constant string of length $n-t$).
\item When $r=2$, $U_{n,r}$ is of the form $\sigma\epsilon^{n-1}$, and has two possible subsequences, namely $\sigma\epsilon^{n-1-t}$ and $\epsilon^{n-t}$.
\item When $t=n-1$ and $r>1$ any subsequence is a single symbol. Since $r>1$, it can be either symbol of the binary alphabet.
\item When $n=r$, $U_{n,r}=C_n$, the binary cyclic string of length $n$. $|D_t(C_n)|$=$d(n,t)$ by definition.
\item In the other cases ($2<r<n$, $t<n-1$), we regard $U^{(1)}_{n,r}$ (``tail first''). We Apply Lemma \ref{l:binSplit} and get $|D_t(U^{(1)}_{n,r})|=|D_t(U_{n-1,r}^{(1)})|+|D_{t+r-n-1}(C_{r-2})| =u(n-1,r,t)+d(r-2,t+r-n-1) $.
\end{itemize}
\end{proof}

\subsection{Solving the recursion}
\begin{theorem}\label{t:UClosedForm} {[Closed form formula for $u(n,r,t)$]}
For all $t<n$, $2 < r\le n$,
\begin{itemize}
\item[(i)] when $r>t$: 
$$
u(n,r,t)=d(r,t)+\sum_{i=t+r-n-1}^{t-2}d(r-2,i).
$$.
\item[(ii)] when $r\le t$:
$$
u(n,r,t)=2+\sum_{i=t+r-n-1}^{r-3}d(r-2,i).
$$.
\end{itemize}
\end{theorem}

\begin{proof}
We sequentially expand $u(n,r,t)$ using Lemma \ref{l:URecursion}, until reaching one of the boundary conditions. After one such expansion we get $u(n,r,t)=u(n-1,r,t)+d(r-2,t+r-n-1)$, after $j$ expansions (assuming the boundary conditions weren't reached) we get $u(n,r,t)=u(n-j,r,t)+\sum_{i=t+r-n-1}^{t+r-n+j-2}d(r-2,i)$. We notice that $i=t+r-n+j-2$ can be negative, and in these cases $d(r-2,i)$ is defined to be zero.
When $r>t$, after $n-r$ steps we get $u(n,r,t)=u(r,r,t)+\sum_{i=t+r-n-1}^{t-2}d(r-2,i)$, and as $u(r,r,t)=d(r,t)$ we get (i) above. When $r\le t$, after $n-t-1$ steps we get $u(n,r,t)=u(t+1,r,t)+\sum_{i=t+r-n-1}^{r-3}d(r-2,i)=2+\sum_{i=t+r-n-1}^{r-3}d(r-2,i)$ and we get (ii) above.
\end{proof}

We notice that when the number of deletions is no greater than $n-r+1$ the expression of $u(n,r,t)$ does not depend on $n$, as stated in the following corollary:
\begin{corollary} 
For $2<r\le n$ and $t  \leq n-r+1$:
\begin{itemize}
\item[(i)] when $r>t$: 
$$
u(n,r,t)=d(r,t)+\sum_{i=0}^{t-2}d(r-2,i).
$$
\item[(ii)] when $r\le t$:
$$
u(n,r,t)=2+\sum_{i=0}^{r-3}d(r-2,i)=1+2^{r-2}.
$$
\end{itemize}
\end{corollary}
\subsection{Improving known lower bounds on number of subsequences}\label{ss:preImprove}
The results of Theorem \ref{t:lower} together with Theorem \ref{t:UClosedForm} lead to the following:

\begin{theorem}\label{t:lowerbound}{[lower bound on the number of subsequences]}
For all $t<n$, $2 < r\le n$ and any $r$-run string $X$
\[
|D_t(X)|\ge d(r,t)+\sum_{i=t+r-n-1}^{\min(t-2,r-3)}d(r-2,i).
\]
\end{theorem}

We compare this result to the previous result by Hirschberg et al. $|D_t(X)|\ge d(r,t) = \sum_{i=0}^t\binom{r-t}i$ \cite{Hirsch}. We limit the comparison to $t \leq r$ as for $t>r$ the previous bound gives 0.

\begin{lemma}
Let $\alpha=t/r$. for $\alpha \in [\frac 13+\frac 1 r,1)$ and for $t\le n-r+1$, $\frac{u(n,r,t)}{d(n,t)}=\Omega\left(\frac{\sqrt{1-\alpha}}{r\alpha}2^{r(\alpha-\frac 1 3)}\right)$.
\end{lemma}
\begin{proof}
$d(r,t)\le(t+1)\max_{i=0}^t\binom{r-t}i$. The series $\binom{r-t}i$ reaches its maximum at $i=\lfloor(r-t)/2\rfloor$. This value is reached, because $t>r/3$ implies that $t>(r-t)/2$. Thus  $d(r,t)\le(t+1)\binom{r-t}{\lfloor(r-t)/2\rfloor}$. Stirling's approximation implies that $\binom{a}{\lfloor a/2\rfloor}=\Theta(\frac{2^a}{\sqrt{a}})$, and thus we get $d(r,t)=O(\frac{t}{\sqrt{r-t}}2^{r-t})$.

On the other hand as $t-2\ge\lfloor\frac{r-2}3\rfloor$, $u(r,t)\ge d(r-2,\lfloor\frac{r-2}3\rfloor)\ge \binom{\lfloor\frac{2}{3}(r-3)\rfloor}{\lfloor\frac{1}{3}(r-3)\rfloor}=\Theta(\frac{2^{\frac{2}{3}(r-3)}}{\sqrt{\frac{2}{3}(r-3)}})=\Theta(\frac{2^{\frac{2}{3}r}}{\sqrt{r}})$, thus  $u(r,t)=\Omega(\frac{2^{\frac{2}{3}r}}{\sqrt{r}})$. \\
$\frac{u(n,r,t)}{d(n,t)}=\Omega(\frac{2^{\frac{2}{3}r} \sqrt{r-t}}{2^{r-t}t\sqrt{r}})$,  thus  $\frac{u(n,r,t)}{d(n,t)}=\Omega(\frac{\sqrt{1-\alpha}}{r\alpha}2^{r(\alpha-\frac 1 3)})$
\end{proof}

For large enough strings ($n>t+r$), the improvement that the bound in Theorem~\ref{t:lowerbound} gives over the result in \cite{Hirsch} depends on the ratio between $r$ and $t$. We depict our improved results in Figure \ref{fig:compareBoundsL}. 

\section{Deletion patterns}
\label{sec:dp}
Consider a string $X$. Deletion of $t$ letters from $X$ can be characterized by partitioning $t$ into the number of letters deleted from each run, leading to the following definition of deletion patterns.

\begin{definition}
Let $X$ be a string s.t. $X=S(x_1,\dots,x_r)$. A deletion pattern of size $t$, is a set of integers $\{y_1,\dots\,y_r\}$ fulfilling $\sum_{i=1}^r y_i = t$ and for all $0\le i \le r$, $y_i \in [0,x_i]$. Each $y_i$ represents the number of letters deleted from the $i$-th run of  $X$. E.g. the deletion pattern $\{2,1,2\}$ for the string $000110000$ results in the subsequence $0100$. let $\mathcal{P}_t(X)$ denote the set of deletion patterns of size $t$ for the string $X$.
\end{definition}

It is important to notice that applying different deletion patterns on a string can result in the same subsequences, E.g. For the string $11011$, the deletion patterns $\{1,1,0\}$ and $\{0,1,1\}$ both result in the subsequence $111$. 
The following lemma ties deletion patters with the study of subsequences (and appears partially in \cite{Bhar}).

\begin{lemma}\label{t:DPSandwitch}
For any  $X=S(x_1,\dots,x_r)$, let $X'$ denote the string $S(x_1-1,\dots,x_r-1)$. Informally $X'$ is the string obtained by deleting one letter from each run in $X$.  It follows that $|\mathcal{P}_t(X')|\le|D_t(X)|\le |\mathcal{P}_t(X)|$.
\end{lemma}
\begin{proof}
Deleting letters from a given string according to a deletion pattern is a deterministic process, and so each deletion pattern yields exactly one subsequence, thus the right inequality follows.
As mentioned before, several deletion patterns can yield the same subsequence, but this redundancy doesn't exist with deletion patterns that preserve the number of the runs in a string, i.e. there isn't a run in which all the symbols are deleted. In this case it is possible to reconstruct the deletion pattern from the subsequence in a unique way, and there is a one-one correspondence between the deletion patterns and the subsequences. The group of deletion patterns of $X$ that preserve the number of runs is exactly the group of deletion patterns in which at least one symbol is not deleted from each run, and is equal to $\mathcal{P}_t(X')$. This group has a one-one correspondence to the subset of $D_t(X)$ of strings with exactly $r$ runs, and thus the left equality holds. 
\end{proof}
\begin{lemma}\label{l:DPSymmetry}
For any $X$, $|\mathcal{P}_t(X)|=|\mathcal{P}_{|X|-t}(X)|$.
\end{lemma}
\begin{proof}
Let $X=S(x_1,\dots,x_r)$ and let $\{y_1,\dots,y_r\}$ be a $t$-deletion pattern.  It follows that $\sum_{i=1}^r y_i = t$.
We define $y_i'=x_i-y_i$ for all $1\le i \le r$. As $y_i\in[0,x_i]$ it follows that $y_i'\in[0,x_i]$ and $\sum_{i=1}^r y_i'=\sum_{i=1}^r (x_i-y_i)=|X|-t$, and so $\{y_i'\}$ is a $(|X|-t)$-deletion pattern of $X$. Each $t$-deletion pattern can be mapped to a $(|X|-t)$-deletion pattern, and this mapping is reversible, thus $|\mathcal{P}_t(X)|=|\mathcal{P}_{|X|-t}(X)|$. 
\end{proof}

\subsection{The number of deletion patterns for balanced strings}
We use the result obtained in Lemma \ref{l:calcP0} and restate it for
deletion patterns to get the following result:
\begin{lemma}\label{l:numDPforB}
$\mathcal |P_t(B_{r,k})|=\sum_{i=0}^{\lfloor\frac{t}{k+1}\rfloor}(-1)^i\binom{r}i\multiset{r}{t-i(k+1)}
= \sum_{i=0}^{\lfloor\frac{t}{k+1}\rfloor}(-1)^i\binom{r}i\binom{r+t-i(k+1)-1}{r-1}$
\end{lemma}

We now study the multiplicative gap between  $|P_t(B_{r,k})|$ and the
previous bounds of \cite{Lev66,Cal69,Hirsch} for values of $t$ {\em
close} to $n/2$ and sufficiently large $r$, $k$.
This is an intriguing setting for $t$ in the context of deletion
channels \cite{Mitz11}.
It follows from basic observations (and also directly from the proof
of Lemma~\ref{l:numDPforB}) that
$$
|P_t(B_{r,k})| \leq \min\left( \multiset{r}{t}, (k+1)^r \right).
$$
The first bound above is exactly that of \cite{Lev66}, while the
second bound follows from the fact that each $y_i$ in a deletion
pattern is an integer between $0$ and $k$ (notice that the former
bound does not depend on the parameter $k$ while the latter does not
depend on $t$).
In what follows, we show that the bound of $(k+1)^r$ improves on the
bounds in \cite{Lev66} and \cite{Cal69,Hirsch} for values of $r$ and
$k$ which are sufficiently large.

For $t=n/2=kr/2$, the bound of $\sum_{i=0}^t \binom{n-t}i$ from
\cite{Hirsch} is exactly $2^{n/2}$.
The bound of  $\binom{r+t-1}{t}$ from \cite{Lev66} is at least
$$
\frac{1}{k}\binom{r(1+k/2)}{r} \geq
\frac{1}{12k\sqrt{r}}\left(\frac{e}{(1+2/k)}\right)^r(1+k/2)^r
$$
Here we use the fact that
{\small{
$$
\binom{r(1+\alpha)}{r} \geq
\frac{1}{12\sqrt{r}} \left(\frac{(r(1+\alpha))^{r(1+\alpha)}}{r^r
(\alpha r)^{\alpha r}}\right) =
\frac{(1+\alpha)^{r}\left(1+1/\alpha\right)^{\alpha r}}{12\sqrt{r}}
$$}}
derived from Stirling's formula; and the fact that for positive $x$,
$(1+1/x)^{x+1} \geq e$.
For $c = \frac{e(1+k/2)}{(k+1)(1+2/k)}$, the above implies that our bound
of $(k+1)^r$ on $|P_t(B_{r,k})|$ is superior to that
given in \cite{Lev66} (and that in \cite{Hirsch}) by a multiplicative factor
of at least
$$
\frac{1}{12k\sqrt{r}}c^r.
$$
Notice that for large $k$, $c>(1+\delta)$ for a constant $\delta >0$.
We conclude that a multiplicative gap of at least that specified above
also holds between $|D_t(B_{r,k})|$ and the bounds in
\cite{Lev66,Cal69,Hirsch}.

For sufficiently small $\epsilon>0$ and $t=n(\frac{1}{2}- \epsilon)$, a similar
analysis will give a gap of $\simeq c^r$ for
$c=\frac{e(1+k/2-\epsilon k)}{(k+1)(1+1/(k/2-\epsilon k))}$.
Here also, for small $\epsilon$ and large $k$;  $c>(1+\delta)$ for a
constant $\delta >0$.
All in all, we get for values $t$ which are {\em close} to $n/2$ and
for sufficiently large $r$ and $k$; that $|P_t(B_{r,k})|$, and thus
our bound of $|D_t(B_{r,k})|$, improves on the bounds  of
\cite{Lev66,Cal69,Hirsch} by an exponential multiplicative factor of
$2^{\Omega(r)}$.


\section{Concluding remarks}
\label{sec:concluding}

In this work we present several operations on binary strings which are monotone with respect to the number of subsequences under deletion.
We show, using the operations studied, that the balanced $r$-run string $B_{r,k}$ and the unbalanced one $U_{n,r}$ obtain the maximum and respectively minimum number of subsequences under deletion.
By devising recursive expressions, we present a precise analysis of the number of subsequences of both $B_{r,k}$ and $U_{n,r}$ under $t$ deletions.
For our lower bound, we quantify our expressions asymptotically.
For our upper bound, we analyze deletion patterns to express our asymptotic improvement over previous bounds.
A direct asymptotic analysis of our expression for $|D_t(B_{r,k})|$ is left open in this work and is subject to future research.


\bibliographystyle{plain}   
\bibliography{yuvref}

\end{document}